\documentclass[runningheads]{llncs}

\usepackage[T1]{fontenc}
\usepackage{graphicx}
\usepackage{amsmath}
\usepackage{amssymb}

\usepackage[ruled,vlined]{algorithm2e}

\usepackage[noend]{algpseudocode}

\SetKwInput{KwData}{Auxiliary}

\newtheorem{prop}{Property}

\begin{document}
\title{A Practical Algorithm for Max-Norm Optimal Binary Labeling of Graphs }
\titlerunning{A Practical Algorithm for Max-Norm Optimization}
%
\author{Filip Malmberg\inst{1}, Alexandre X. Falc\~{a}o\inst{2}}
\authorrunning{F. Malmberg, A. X. Falc\~{a}o}
%
\institute{Centre for Image Analysis, Dept. of Information Technology, Uppsala University, Sweden\\ \email{filip.malmberg@it.uu.se} \and Institute of Computing, University of Campinas, Brazil\\ \email{afalcao@ic.unicamp.br}}

\maketitle              
\begin{abstract}
This paper concerns the efficient implementation of a method for optimal binary labeling of graph vertices, originally proposed by Malmberg and Ciesielski (2020). This method finds, in quadratic time with respect to graph size, a labeling that globally minimizes an objective function based on the $L_\infty$-norm. The method enables global optimization for a novel class of optimization problems, with high relevance in application areas such as image processing and computer vision. In the original formulation, the Malmberg-Ciesielski algorithm is unfortunately very computationally expensive, limiting its utility in practical applications. Here, we present a modified version of the algorithm that exploits redundancies in the original method to reduce computation time. While our proposed method has the same theoretical asymptotic time complexity, we demonstrate that is substantially more efficient in practice. Even for small problems, we observe a speedup of 4-5 orders of magnitude. This reduction in computation time makes the Malmberg-Ciesielski method a viable option for many practical applications. 

\end{abstract}
\section{Introduction}
Many problems in computer science and pattern recognition can be as finding vertex labeling of a graph, such that the labeling optimizes some application-motivated objective function. In their recent work, Malmberg and Ciesielski~\cite{malmberg2020two} proposed a quadratic time algorithm for assigning binary labels to the vertices of a graph, such that the resulting labeling is optimal according to an objective function based on the max-norm, or $L_\infty$ norm. Here, we consider the efficient implementation of the algorithm proposed by Malmberg and Ciesielski. We present a version of their algorithm that, while having the same quadratic asymptotic time complexity, is orders of magnitude faster in practice. 

A key part of the Malmberg-Ciesielski algorithm is to solve a sequence of \emph{Boolean 2-satisfiability} (2-SAT) problems. Malmberg and Ciesielski observe that each such 2-SAT problem can be solved in linear time using, e.g., Aspvall's algorithm~\cite{aspvall1979linear}. They also observe, however, that there is a high degree of similarity between each consecutive 2-SAT problem in the sequence and that solving each 2-SAT problem in isolation thus appears inefficient. Here, we show that this redundancy between subsequent 2-SAT problems can indeed be exploited to formulate a substantially more efficient version of the algorithm.

\section{Background and motivation}
\label{sec:background}
We consider the problem of assigning a binary label ($0$ or $1$) to a set of variables identified by indices $1,\ldots,n$. A canonical problem is to find a binary labeling $\ell: [1,n]\rightarrow \{0,1\}$ that minimizes an objective function of the form

\begin{equation}
\label{eq:pnorm}
E_p(\ell) : = \sum_{i } \phi_i^p(\ell(i))+\sum_{(i,j) \in \mathcal{N}} \phi^p_{ij}(\ell(i),\ell(j)) ,
\end{equation}

\noindent where $\ell(i) \in \{0,1\}$ denotes the label of variable $i$ and $\mathcal{N}$ is a set of pairs of variables that are considered \emph{adjacent}. 

The functions $\phi_i(\cdot)$ are referred to as \emph{unary} terms. Each unary term depends only on the value of a single binary variable, and they are used to indicate the preference of an individual variable to be assigned each particular label. 

The functions $\phi_{ij}(\cdot,\cdot)$ are referred to as \emph{pairwise} terms. Each pairwise term depends on the labels assigned to two variables simultaneously, and thus introduces a dependency between the labels assigned to the variables. Typically, this dependency between variables is used to express that the desired solution should have some degree of smoothness, or regularity.

As established by  Kolmogorov and Zabih~\cite{kolmogorov2004energy}, the labeling problem described above can be solved to global optimality under the condition that all pairwise terms are submodular, which in the form presented here means that they must satisfy the inequality

\begin{equation}
\label{eq:p_submodularity}
\phi_{ij}^p(0,0)+\phi_{ij}^p(1,1) \leq \phi_{ij}^p(0,1)+\phi_{ij}^p(1,0).
\end{equation}

If the problem contains non-submodular binary terms, finding a globally optimal labeling is known to be NP-hard in the general case~\cite{kolmogorov2004energy}. Practitioners looking to solve such optimization problems must therefore first verify that their local cost functional satisfies the appropriate submodularity conditions. If this is not the case, they must resort to approximate optimization methods that may or may not produce satisfactory results for a given problem instance~\cite{kolmogorov2007minimizing}. Recently, however, Malmberg and Ciesielski~\cite{malmberg2020two} showed that in the limit case, as $p$ approaches to infinity, the requirement for submodularity disappears!  To characterize the labelings that minimize~\ref{eq:pnorm} as $p$ goes to infinity, we first observe that as $p$ goes to infinity the objective function $E_p$  itself converges to 

\begin{equation}
\label{eq:maxnorm}
E_\infty(\ell) : =  \max\bigl\{\max_{i} \phi_i(\ell(i)), 
\max_{( i,j) \in \mathcal{N} } \phi_{ij}(\ell(i),\ell(j))\bigr\}.\!\!\!\!\!
\end{equation}

\noindent i.e., the objective function becomes the max-norm of the vector containing all unary and pairwise terms.  A more refined way of characterizing the solution is the framework of \emph{lexicographic max-ordering} (Lex-MO)~\cite{ehrgott1995lexicographic,Ehrgott1999,ehrgott2005multicriteria}. The same concept was also studied by Levi and Zorin, who used the term \emph{strict minimizers}~\cite{levi2014strict}.  In this framework, two solutions are compared by
ordering all elements (in our case, the values of all unary and pairwise terms for a given solution) non-increasingly and then performing their lexicographical comparison. This avoids the potential drawback of the $E_\infty$ objective function, that it does not distinguish between solutions with high or low errors below the maximum error. The Malmberg-Ciesielski algorithm~\cite{malmberg2020two} computes, in polynomial time, a labeling that globally minimizes $E_\infty$, even in the presence of non-submodular pairwise terms. Under certain conditions, the same algorithm is also guaranteed to produce a solution that is optimal in the Lex-MO sense.  

\section{Preliminaries}
\label{sec:preliminaries}
In this section, we recall briefly the Malmberg-Ciesielski algorithm, along with some concepts needed for exposition of our proposed efficient implementation of this algorithm in Section~\ref{sec:method}.

\subsection{Boolean 2-satisfiability} 
We start by recalling the Boolean 2-satisfiability (\emph{2-SAT}) problem. Given a set of Boolean variables $\{x_1, \ldots, x_n\}$, $x_i \in \{0,1\}$ and a set of logical constraints on pairs of these variables, the 2-SAT problem consists of determining whether it is possible to assign values to the variables so that all the constraints are satisfied (and to find such an assignment, if it exists). To formally define the 2-SAT problem, we say that a \emph{literal} is either a Boolean variable $x$ or its negation $\neg x$. A 2-SAT problem can then be defined in terms of a Boolean expression that is a conjunction of \emph{clauses}, where each clause is a disjunction of two literals. Expressions on this form are known as 2-CNF formulas, where CNF stands for \emph{conjunctive normal form}. The 2-SAT problem consists of determining if there exists a truth assignment to the variables involved in a given 2-CNF formula that makes the whole formula true. If such an assignment exists, the 2-SAT problem is said to be \emph{satisfiable}, otherwise it is \emph{unsatisfiable}. As an example, the following expression is a 2-CNF formula involving three variables $x_1, x_2, x_3$, and two clauses:

\begin{equation}
(x_1 \lor x_2)\land(x_2\lor \neg x_3)
\end{equation}

\noindent This example formula evaluates to \emph{true} if we, e.g., assign all three variables the value $1$ (or \emph{true}). Thus the 2-SAT problem represented by this 2-CNF formula is satisfiable. 

For any 2-CNF formula, the 2-SAT problem is solvable in linear time w.r.t to the number of clauses\footnote{This is in contrast to the general Boolean satisfiability problem, where clauses are allowed to contain more than two literals. Already the 3SAT problem, where each clause can have at most three literals, is NP-hard. } using, e.g., Aspvall's algorithm~\cite{aspvall1979linear}.

We now introduce some further notions related to 2-SAT problems needed for our exposition, using the convention that $x_i$ and $\neg x_i$ denote literals, while $v_i$ denotes a literal whose truth value is unknown and $\bar{v_i}$ is its complementing literal. 

Every clause $(v_i \lor v_j)$ in a 2-CNF formula is logically equivalent to an implication from one of its variables to the other: 

\begin{equation}
 (v_i \lor v_j) \equiv (\bar{v_i} \Rightarrow v_j) \equiv (\bar{v_j} \Rightarrow v_i) \;.
\end{equation}

\noindent As established by Aspvall et al.~\cite{aspvall1979linear}, this means that every 2-SAT problem $F$ can be associated with an \emph{implication graph} $G_F=(V,E)$, a directed graph with vertices $V$ and edges $E$ constructed as follows:

\begin{enumerate}
\item For each variable $x_i$, we add two vertices named $x_i$ and $\neg{x_i}$ to $G_F$. The vertices  $x_i$ and $\neg{x_i}$ are said to be \emph{complementing}.
\item For each clause $(v_i \lor v_j)$ of $F$, we add edges $(\bar{v_i}, v_j)$ and $(\bar{v_i}, v_j)$ to $G_F$. 
\end{enumerate} 

\noindent Each vertex in the implication graph can thus be uniquely identified with a literal, and each edge identified with an implication from one literal to another. We will therefore sometimes interchangeably refer to a vertex in the implication graph by its corresponding literal $v_i$. For a given truth assignment, we say that a vertex in the implication graph \emph{agrees} with the assignment if the corresponding literal evaluates to \emph{true} in the assignment. The implication graph $G_F$ is \emph{skew symmetric} in the sense that if $(v_i, v_j)$  is an edge in $G_F$, then $(\bar{v_i}, \bar{v_j})$ is also an edge in $G_F$. We observe that it follows that for every path $\pi=(v_1, v_2, \ldots, v_k)$ in $G_F$, the path $\bar\pi=(\bar{v}_k, \bar{v}_{k-1}, \ldots , \bar {v}_1)$ is also a path in $G_F$.

In proving the correctness of our proposed algorithm, we will rely on the following property which is due to Aspvall et al.~\cite{aspvall1979linear}:
\begin{prop}
\label{prop1}
A given truth assignment satisfies a formula $F$ if and only if there is no vertex in $G_F$ for which the corresponding literal agrees with the assignment, with an outgoing edge to a vertex not agreeing with the assignment.
\end{prop}

\subsection{The Malmberg-Ciesielski algorithm} 
For a complete description of the Malmberg-Ciesieleski algorithm, we refer the reader to the original publication (\cite{malmberg2020two}, Algorithm 1). We focus here on a key aspect of the algorithm, which is to solve a sequence of 2-SAT problems. In this step, we identify the variables to be labeled with the Boolean variables involved in a 2-SAT problem. A truth assignment $T$ for the Boolean variables naturally translates to a labeling $\ell$. For this step of the algorithm, we are given an ordered sequence $\mathcal{C}$ of clauses, ordered by a priority derived from the unary and pairwise terms in Eq.~\ref{eq:maxnorm}. Informally, the algorithm operates as follows:

\begin{itemize}
\item Initialize $F$ to be an empty 2-CNF formula, containing no clauses.
\item For each clause $c$ in $\mathcal{C}$, in order:
\begin{itemize}
\item If $F \land c$ is satisfiable, then set $F \leftarrow F \land c$.
\end{itemize}
\end{itemize}
 
\noindent At all steps of the above algorithm, the formula $F$ remains satisfiable. At the termination of the algorithm, the formula $F$ defines a unique truth assignment $T$ and therefore also a labeling $\ell$. For the specific sequence $\mathcal{C}$ of clauses defined by Malmberg and Ciesieleski, the resulting labeling is guaranteed to globally minimize the objective function in Eq.~\ref{eq:maxnorm}.

In each iteration, we need to determine if $F\land c$ is satisfiable, i.e., solve the 2-SAT problem associated with the formula $F\land c$. Malmberg and Ciesieleski suggest to use Aspvall's algorithm for this purpose, with an asymptotic time complexity of $\mathcal{O}(|F|) \leq \mathcal{O} (|\mathcal{C}|)$. Let $N=n+|\mathcal{N}|$ denote the total number of unary and pairwise terms in Eq.~\ref{eq:maxnorm}. By its design, the number of clauses in the sequence $\mathcal{C}$ is $\mathcal{O}(N)$, leading to the asymptotic time complexity of  $\mathcal{O}(N^2)$ for the Malmberg-Ciesieleski algorithm implemented using Aspvall's algorithm.

\section{Proposed algorithm}
\label{sec:method}

As observed in the previous section, the Malmberg-Ciesieleski algorithm iteratively builds a formula $F$ that remains satisfiable at each step of the algorithm. Our approach for improving the efficiency of the computations is to maintain, at each step of the algorithm, a truth assignment that satisfies the current formula $F$. When trying to determine whether the next clause $c$ in the sequence $C$ can be appended to $F$ without rendering the formula unsatisfiable, we show that this previous truth assignment can be utilized to reduce the computation time. We represent a truth assignment $T$ to the Boolean variables of a 2-SAT problem as a function $T:[1,n]\rightarrow\{0,1\}$, so that $T(i)$ is the value assigned to variable $x_i$. Trivially, if $T$ satisfies $c$ then is also satisfies $F \land c$, so we focus on the case where $T$ does not satisfy the next clause $c$.

We will consider 2-SAT-solving under \emph{assumptions}~\cite{een2003extensible}, i.e., given a satisfiable formula, we ask if the same formula still satisfiable if we assume given values for a subset of the variables? Such assumptions will be represented by a set of vertices in the implication graph -- since each vertex corresponds to a literal, the set of vertices corresponds to a set of literals that are all assumed to evaluate to \emph{true}. We assume that vertex sets used in this context are internally conflict-free, i.e., they do not contain both a vertex and its complement.

Below we will present an efficient algorithm for solving a 2-SAT problem under a set of assumptions $A$, given a truth assignment $T$ that satisfies the formula \emph{without} the assumptions. To see how such a procedure helps us in efficiently implementing the Malmberg-Ciesielski algorithm, we observe that by De Morgan's laws a clause $(v_i \lor v_j)$ can be rewritten as $\neg(\bar{v_i} \land \bar{v_j})$. In this form, it is easier to see that in order to satisfy this clause, the truth assignment $T$ must satisfy exactly one of the expressions $({v_i} \land {v_j})$, $({v_i} \land \bar{v_j})$,  or $(\bar{v_i} \land {v_j})$.  Each of these expressions represent a set of assumptions, and therefore $F\land(v_i \lor v_j)$ is satisfiable if and only if $F$ is satisfiable under one of the following sets of assumptions $A$:  $\{{v_i} , {v_j}\}$, $\{ {v_i} , \bar{v_j}\}$, or $\{\bar{v_i} , {v_j}\}$. We note also that in the special case that $i=j$, the above argument can be simplified further. In this case, the formula reduces to $F\land(v_i)$ which is equivalent to solving $F$ under the assumption $A =\{v_i\}$. 

The procedure listed in Algorithm~\ref{addifsat} utilizes this result to perform the inner loop of the Malmberg-Ciesieleski algorithm: It determines whether a given clause can be added to a satisfiable formula without making it unsatisfiable. If so, it updates an implication graph representing the formula to include the new clause. Algorithm~\ref{addifsat} utilizes a procedure \emph{SolveWithAssumptions}, which we will now describe.

\LinesNumbered
\SetEndCharOfAlgoLine{}
\begin{algorithm}[tb]
\SetKw{remove}{remove}
\SetKw{from}{from}

\KwIn{An implication graph $G$ representing a 2-SAT problem. A clause $c=(v_i)\lor(v_j)$. A truth assignment $T$ that satisfies the formula $F$ encoded by $G$. }
  
\KwResult{A truth value indicating if $F\land c$ is satisfiable. If it is, then $T$ is a truth assignment satisfying $F \land c$ and $G$ encodes $F \land c$. Otherwise, $T$ and $G$ are unmodified. }

\BlankLine 
Set \emph{satisfiable} $\leftarrow$ \emph{false}\;
\If{$T$ satisfies $c$}{
	Set \emph{satisfiable} $\leftarrow$ \emph{true}\;
}
\Else{
\If{$v_i=v_j$}{
\If{SolveWithAssumptions($G$,$\{ v_i\}$,$T$)}{
	Set \emph{satisfiable} $\leftarrow$ \emph{true}\;
}
}
\Else{ \tcc{$v_i\neq v_j$}
\If{SolveWithAssumptions($G$,$\{ {v_i}, v_j\}$,$T$)}{
	Set \emph{satisfiable} $\leftarrow$ \emph{true}\;
}

\ElseIf{SolveWithAssumptions($G$,$\{ \bar{v_i}, v_j\}$,$T$)}{
	Set \emph{satisfiable} $\leftarrow$ \emph{true}\;
}

\ElseIf{SolveWithAssumptions($G$,$\{ v_i, \bar{v_j}\}$,$T$)}{
	Set \emph{satisfiable} $\leftarrow$ \emph{true}\;
}
}
}

\If{$satisfiable$}{
	Add edges $(\bar{v_i}, v_j)$ and $(\bar{v_j}, v_i)$ to $G$\;
}

Return \emph{satisfiable}

\caption{CheckSolvable($G$,$C$,$T$)} \label{addifsat}
\end{algorithm}

\LinesNumbered
\SetEndCharOfAlgoLine{}
\begin{algorithm}[tb]
\SetKw{remove}{remove}
\SetKw{from}{from}

\KwIn{An implication graph $G$ representing a 2-SAT problem. A set of assumtions $A$, without internal conflicts. A truth assignment $T$ that satisfies the formula $F$ encoded by $G$. }
  
\KwResult{A truth value indicating the existence of a truth assignment $T'$ that satisfies the formula $F$ encoded by $G$ while simultaneously satisfying the assumptions $A$. If the algorithm returns \emph{true}, then $T$ is a truth assignment satisfying this criterion. Otherwise, $T$ is unmodified. }

\KwData{ A FIFO (or LIFO) queue $Q$ of vertices; A set of vertices $C$. }  

Set $C \leftarrow \emptyset$

\ForEach{$v \in A$}{
	Insert $v$ in $Q$\;
	Insert $v$ in $C$\; 
}

\While{$Q$ is not empty}{
	Pop a vertex $v$ from $Q$\;
	\If{$v$ disagrees with $T$}{
		\ForEach{vertex $w$ such $v$ has an outgoing edge to $w$ }{
			\If{$\bar{w}\in C$}{
				Return \emph{false} and exit\;
			}
			\ElseIf{$w\notin C$ }{
				Insert $w$ in $Q$\;
				Insert $w$ in $C$\; 
			}

		}
	}
}
\ForEach{vertex $v \in C$}{
	Set value of $T$ for the variable corresponding to $v$ so that it agrees with $v$.\;
}
Return \emph{true}
\caption{SolveWithAssumptions($G$,$A$,$T$)} \label{solve_under_assumptions}
\end{algorithm}

Let $F$ be a formula with corresponding implication graph $G_F = (V,E)$, let $T$ be a truth assignment for the variables associated with $F$, and let $A$ be a set of assumptions. We define $R_{A,T}\subseteq{V}$ as the set of vertices that are reachable in $G_F$ from any vertex in $A$ without traversing an edge that is outgoing from a vertex that agrees with $T$. The main theoretical result that enables our proposed algorithm is summarized in the following theorem:

\begin{theorem}
Assume that $F$ is satisfiable. Let $T$ be a truth assignment that satisfies $F$, and let $A$ be a set of assumptions. Then $F$ is satisfiable under the assumptions A if and only if the subgraph $R_{A,T}$ does not contain a pair of complementing vertices.
\end{theorem}
\begin{proof}
For the first part of the proof, assume that $R_{A,T}$ does contain a pair of complementing vertices $v_i$ and $\bar{v_i}$. Then the assumptions $A$ directly imply that both $v_i$ and  $\bar{v_i}$ are simultaneously satisfied, which is clearly a contradiction, and so $F$ is not satisfiable under the assumptions $A$.

For the second part of the proof, assume that $R_{A,T}$ does not contain any pair of complementing vertices. We may then construct a well-defined truth assignment $T'$ from the given truth assignment $T$ by setting, for every vertex in $R_{A,T}$, the correponding variable to the corresponding truth value. For any vertex $v_i \notin R_{A,T}$, we have $T(i)=T'(i)$. Furthermore, the truth assignment $T'$ agrees with all assumptions in $A$.  

Next assume, with the intent of constructing a proof by contradiction, that the truth assignment $T'$ constructed above does not satisfy $F$. Then by Property~\ref{prop1} there exists at least one vertex $v_i$ agreeing with $T'$ that has an outgoing edge to a vertex $v_j$ not agreeing with $T'$. We now consider all four possibilities for the truth assignment $T$ with respect to the variables corresponding to $v_i$ and $v_j$ :

\begin{enumerate}
\item Assume that both $v_i$ and $v_j$ agree with $T$. Then since $v_j$ does not agree with $T'$ we must have $\bar{v_j}\in R_{A,T}$, i.e., there exists a path $\pi$ from $A$ to $\bar{v_j}$ that does not traverse an edge outgoing from a vertex that agrees with $T$. By the skew symmetry of the implication graph, there is an outgoing edge from $\bar{v_j}$ to $\bar{v_i}$, and we may thus append this edge to the path $\pi$ to see that $\bar{v_i}$ is also in $R_{A,T}$, contradicting that $v_i$ agrees with $T'$.  Thus, the assumption that both $v_i$ and $v_j$ agree with $T$ leads to a contradiction.
\item Assume that $v_i$ agrees with $T$ but $v_j$ does not. Since $v_i$ has an outgoing edge to $v_j$, this contradicts that $T$ satisfies $F$, and so the assumption that $v_i$ agrees with $T$ but $v_j$ does not agree with $T$ leads to a contradiction.
\item Assume that $v_j$ agrees with $T$ but $v_i$ does not. Then $v_i$ and $\bar{v_j}$ are both in $R_{A,T}$. There is an outgoing edge from $v_i$ to $v_j$, and $v_i$ disagrees with $T$, and thus $v_j$ is also in $R_{A,T}$, contradicting the assumption that $R_{A,T}$ does not contain both a vertex and its complement. Thus, the assumption that $v_j$ agrees with $T$ but $v_i$ does not agree with $T$ leads to a contradiction.
\item Assume that neither $v_i$ nor $v_j$ agree with $T$. Then since $v_i$ agrees with $T'$ we must have $v_i \in R_{A,T}$, i.e.,  there exists a path $\pi$ from $A$ to $v_i$ that does not traverse an edge outgoing from a vertex that agrees with $T$. But since there is an outgoing edge from $v_i$ to $v_j$ and $v_i$ does not agree with $T$, we may append $\pi$ with this edge to see that $v_j$ must also be in $R_{A,T}$, contradicting that $v_j$ disagrees with $T'$. Thus, the assumption that neither $v_i$ nor $v_j$ agree with $T$ leads to a contradiction
\end{enumerate}

The four cases above cover all possible configurations for the thuth values of the variables corresponding to $v_i$ and $v_j$ in the truth assignment $T$, and each case leads to a contradiction. We conclude that the assumption that $T'$ does not satisfy $F$ leads to a contradiction, and thus $T'$ must satisfy $F$. This completes the proof. \qed
\end{proof}

Based on the theorem presented above, we can solve a 2-SAT problem under given assumptions if we can find the set $R_{A,T}$. We observe that for a given set of assumptions, the set $R_{A,T}$ can easily be found in $\mathcal{O}(V+E)$ time using, e.g., breadth-first search. If we, during this breadth-first search, encounter a vertex whose complement is already confirmed to be in $R_{A,T}$, we may terminate the search and return \emph{false}. Pseudocode for this approach is presented in Algorithm~\ref{solve_under_assumptions}. With an upper bound of $\mathcal{O}(V+E)$ for solving each 2-SAT problem, the proposed approach has the same asymptotic time complexity as the approach using Aspvall's algorithm. In practice, however, we will see that the set $R_{A,T}$ is a very small subset of the implication graph, making this approach much faster than running Aspvall's algorithm for every iteration of the Malmberg-Ciesielski algorithm.

\section{Evaluation}
To evaluate the performance of our proposed version of the Malmberg-Ciesielski to the original formulation using Aspvall's algorithm, perform an empirical study emulating a typical optimization scenario in image processing and computer vision. We perform binary labeling of the pixels of a 2D image of size $W \times H$. The neighborhood relation $\mathcal{N}$ is defined by the standard 4-connectivity used in image processing. Values for the unary and pairwise terms are drawn randomly from a uniform distribution. We then compare the computation time of the two implementations, for image sizes varying from $8\times 8$ to $64\times 64$. We only measure the time required for solving the sequence of 2-SAT problems, as this is the only aspect that differs between the implementations. The results are shown in Figure~\ref{fig:comparison}. As the figure shows, the computation time for the implementation based on Aspvall's algorithm increases dramatically with increasing problem size. For an image of size $64\times 64$, the implementation based on Aspvall's algorithm runs in 62 seconds, while the proposed implementation only requires 0.004 seconds for the same computation -- a speedup of more than four orders of magnitude.

To further study the computation time of the proposed implementation with respect to problem size, we perform a separate experiment on images with sizes varying from $128 \times 128$ to $4096 \times 4096$, for which the implementation using Aspvall's algorithm becomes prohibitively slow. The results are shown in Figure~\ref{fig:timing}. As can be seen from the figure the empirical relation between problem size and computation time appears closer to a linear function across this range, rather than quadratic relation suggested by the worst-case asymptotic time complexity.

\begin{figure}[t]
\centering
\includegraphics[width=0.63\textwidth]{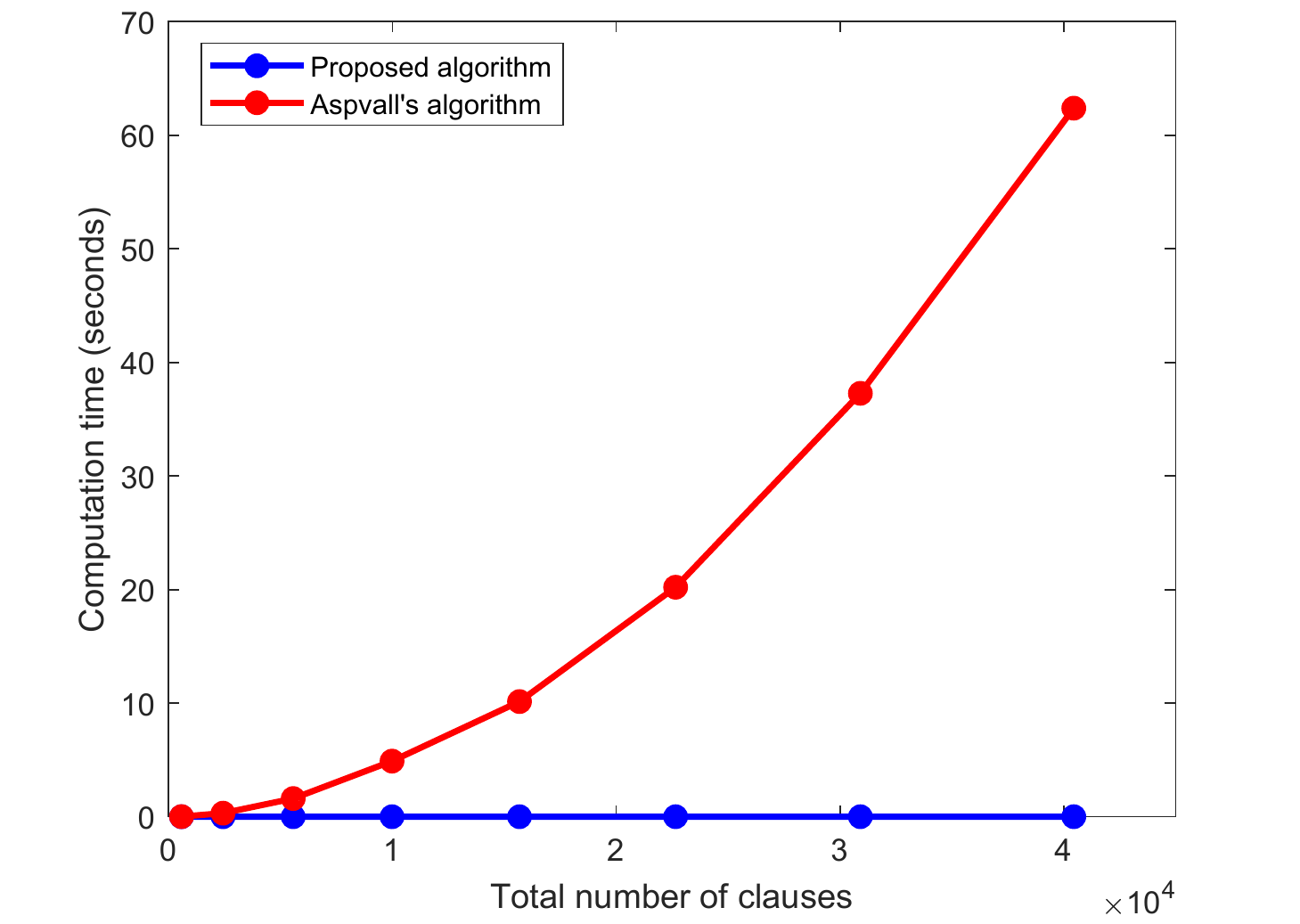}

\caption{Comparison of computation time between the proposed implementation of the Malmberg-Ciesielski method, and the original formulation using Aspvall's algorithm, with respect to the total number of clauses in the 2-SAT sequence.}
\label{fig:comparison}
\end{figure}

\begin{figure}[t]
\centering
\includegraphics[width=0.63\textwidth]{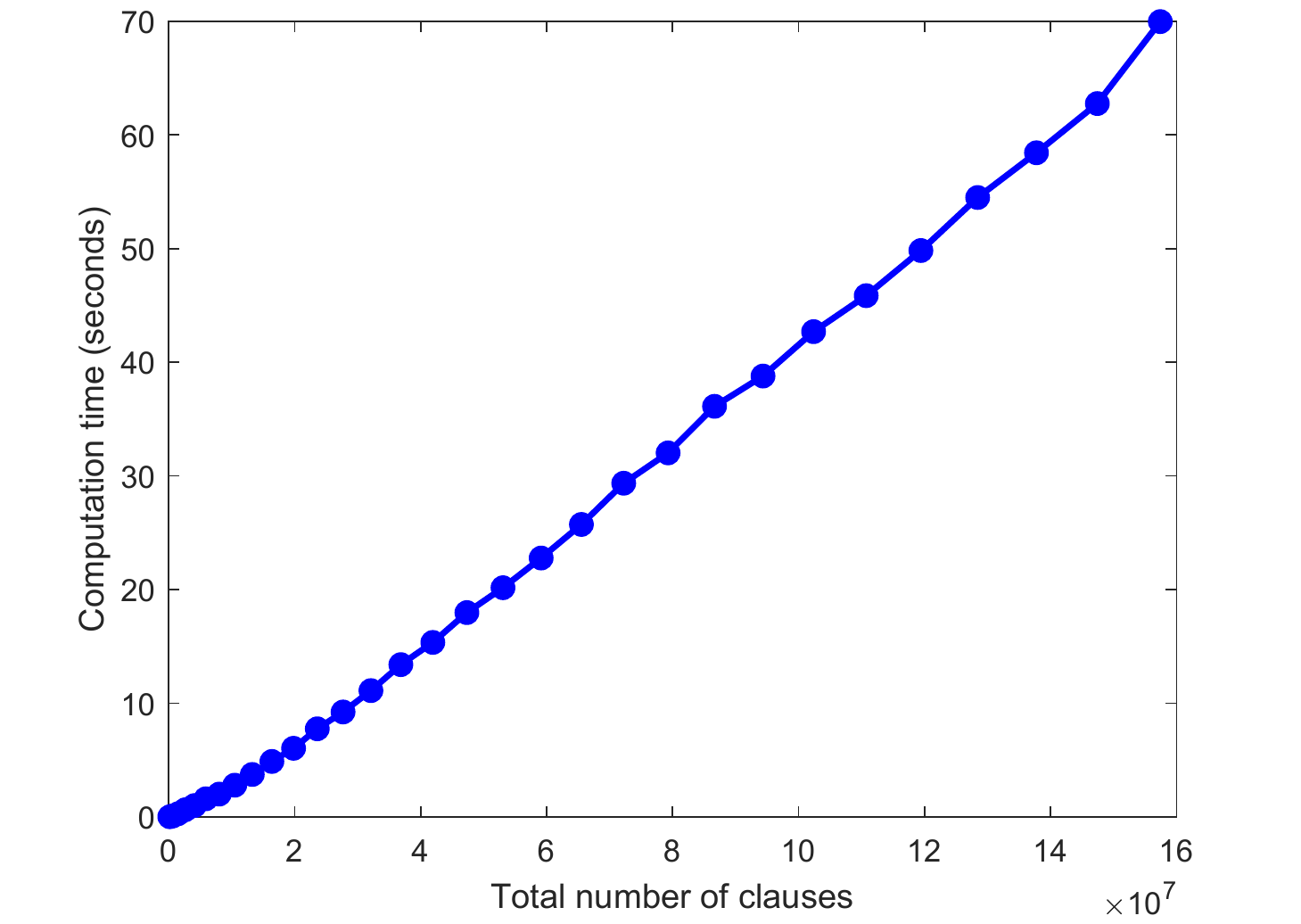}
\caption{Computation time of the proposed implementation in relation to problem size.}
\label{fig:timing}
\end{figure}

\section{Conclusions}
We have proposed a modified, efficient implementation of the Malmberg-Ciesielski method for optimal binary labeling of graphs. While our proposed implementation has the same asymptotic run-time complexity as the original algorithm, we demonstrate that it is orders of magnitude faster in practice. This reduction in computation time makes the Malmberg-Ciesielski method a viable option for many practical applications. 

\subsubsection*{Acknowledgements}
This work was supported by a SPRINT grant (2019/08759-2) from the S\~{a}o Paulo Research Foundation (FAPESP) and Uppsala University. 

%

%
%
%
\bibliographystyle{splncs04}
\bibliography{refs}

\begin{thebibliography}{1}
\providecommand{\url}[1]{\texttt{#1}}
\providecommand{\urlprefix}{URL }
\providecommand{\doi}[1]{https://doi.org/#1}

\bibitem{aspvall1979linear}
Aspvall, B., Plass, M.F., Tarjan, R.E.: A linear-time algorithm for testing the
  truth of certain quantified boolean formulas. Inf. Process. Lett.
  \textbf{8}(3),  121--123 (1979)

\bibitem{een2003extensible}
E{\'e}n, N., S{\"o}rensson, N.: An extensible {SAT}-solver. In: International
  conference on theory and applications of satisfiability testing. pp.
  502--518. Springer (2003)

\bibitem{ehrgott1995lexicographic}
Ehrgott, M.: Lexicographic max-ordering-a solution concept for multicriteria
  combinatorial optimization. Deutsche Nationalbibliothek  (1995)

\bibitem{Ehrgott1999}
Ehrgott, M.: A characterization of lexicographic max-ordering solutions (1999)

\bibitem{ehrgott2005multicriteria}
Ehrgott, M.: Multicriteria optimization, vol.~491. Springer Science \& Business
  Media (2005)

\bibitem{kolmogorov2007minimizing}
Kolmogorov, V., Rother, C.: Minimizing nonsubmodular functions with graph
  cuts-a review. IEEE TPAMI
   \textbf{29}(7) (2007)

\bibitem{kolmogorov2004energy}
Kolmogorov, V., Zabih, R.: What energy functions can be minimized via graph
  cuts? IEEE TPAMI
  \textbf{26}(2),  147--159 (2004)

\bibitem{levi2014strict}
Levi, Z., Zorin, D.: Strict minimizers for geometric optimization. ACM TOG  \textbf{33}(6), ~185 (2014)

\bibitem{malmberg2020two}
Malmberg, F., Ciesielski, K.C.: Two polynomial time graph labeling algorithms
  optimizing max-norm-based objective functions. Springer JMIV  \textbf{62}(5),  737--750 (2020)

\end{thebibliography}

\end{document}